\documentclass{fundam}

\usepackage{url} 
\usepackage[ruled,lined]{algorithm2e}
\usepackage{graphicx}
\usepackage{tikz}
\usetikzlibrary{automata, positioning, arrows}
\usepackage{amssymb}
\usepackage{amsmath}
\usepackage{makecell}
\usepackage[title]{appendix}
\usepackage{xcolor}
\definecolor{DarkGreen}{RGB}{0,127,64}

\newtheorem{thm}{Theorem}

\newtheorem{prop}[thm]{Proposition}

 \newtheorem{defin}[thm]{Definition}
 \newtheorem{rem}[thm]{Remark}

\newcommand{\Z}{\mathbb{Z}}

\newcommand{\N}{\mathbb N}
\newcommand{\K}{\mathbb K}

\newcommand{\rp}{{r_P}}
\newcommand{\rqq}{{r_Q}}

\newcommand{\ov}[1]{\overline{#1}}

\def\al{\operatorname{\alpha}}
\def\be{\operatorname{\beta}}

\newenvironment{expl}{\paragraph{Explanation:}}{\hfill$\square$}

  \usepackage[colorlinks = true, linkcolor = black, urlcolor = black, citecolor = black, anchorcolor = black]{hyperref}

\begin{document}

\setcounter{page}{107}
\publyear{2021}
\papernumber{2094}
\volume{184}
\issue{2}

    \finalVersionForIOS

\title{High-degree Compression Functions on Alternative Models of Elliptic Curves and their Applications}

\author{Micha\l{} Wro\'nski\thanks{Address  for correspondence: Kaliskiego 2, 00-908 Warsaw, Poland. \newline \newline
          \vspace*{-6mm}{\scriptsize{Received November 2021; \ revised December 2021.}}}, Tomasz Kijko, Robert Dry\l{}o
\\
Faculty of Cybernetics \\
Military University of Technology in Warsaw \\
Kaliskiego 2, 00-908 Warsaw, Poland\\
\{michal.wronski, tomasz.kijko, robert.drylo\}@wat.edu.pl
 }

\maketitle

\runninghead{M.  Wro\'nski et al.}{High-degree Compression Functions on Alternative Models of Elliptic Curves...}

\begin{abstract}
This paper presents method for obtaining high-degree compression functions using natural symmetries in a given model of an elliptic curve. Such symmetries may be found using symmetry of involution $[-1]$ and symmetry of translation morphism $\tau_T=P+T$, where $T$ is the $n$-torsion point which naturally belongs to the $E(\mathbb K)$ for a given elliptic curve model. We will study alternative models of elliptic curves with  points of order $2$ and $4$, and specifically Huff's curves and the Hessian family of elliptic curves (like Hessian, twisted Hessian and generalized Hessian curves) with a point of order $3$.
We bring up some known compression functions on those models and present new ones as well. For (almost) every presented compression function, differential addition and point doubling formulas are shown. As in the case of high-degree compression functions manual investigation of differential addition and doubling formulas is very difficult, we came up with a Magma program which relies on the Gr\"obner basis.
We prove that if for a model $E$ of an elliptic curve exists an isomorphism $\phi:E \to E_M$, where $E_M$ is the Montgomery curve and for any $P \in E(\mathbb K)$ holds that $\phi(P)=(\phi_x(P), \phi_y(P))$, then for a model $E$ one may find compression function of degree $2$. Moreover, one may find, defined for this compression function, differential addition and doubling formulas of the same efficiency as Montgomery's.
However, it seems that for the family of elliptic curves having a natural point of order $3$, compression functions of the same efficiency do not exist.
\end{abstract}


\section{Introduction}

Elliptic curve cryptography has been evolving over the years. Classical elliptic curve cryptography (ECC) algorithms, such as ECDH, have been replaced by the isogeny-based cryptography algorithms, such as SIDH \cite{Jao11}, SIKE \cite{Aza20} and CSIDH \cite{Cas18}. Isogeny-based cryptography is believed to be resistant to attacks by quantum computers, unlike classical solutions. It is worth noteing, that although many current ECC notions differ from those relied upon years ago, $x$-line arithmetic proposed by Peter L. Montgomery in \cite{Mon87} is still widely used, especially in isogeny-based cryptography. For example, in SIKE reference implementation, $x$-line arithmetic on Montgomery curves is applied, using $XZ$ coordinates.

Over the last 20 years, numerous alternative elliptic curves models have been proposed, e.g. Edwards \cite{Edw07}, \cite{Ber07}, twisted Edwards \cite{Ber08}, Hessian \cite{Joy01}, twisted Hessian \cite{Ber15}, generalized Hessian \cite{Far10}, Huff's and generalized Huff's \cite{Joy10} curves, and many others. However, efficient $x$-line arithmetic has not been proposed for all of these alternative elliptic curves models.

The more general concept of $x$-line arithmetic, especially in application to elliptic curves cryptography is compression function, which is described in details in Section \ref{sec2}.  In general, compression functions are well-known methods of obtaining shorter representation of elements used in many cryptographic applications. The basic example is representation of point on Weierstrass curve or Montgomery curve using only its $x$-coordinate. What is important, this concept may be extended to other, alternative models of elliptic curves, as same as for representation of finite fields elements. It is worth noting that in XTR \cite{XTR00} algorithm, instead of using full representation of element $h \in \mathbb F_{p^6}$ from subgroup of order $p^2-p+1$, it is enough to use the trace $Tr(h)$ which is defined over $\mathbb F_{p^2}$.
Nowadays, the most important application of compression functions is using them in isogeny-based cryptography, especially in SIDH \cite{Feo14}, SIKE \cite{SIKE}, application of Velusqrt \cite{Ber20} method to CSIDH, CSURF and other algorithms.

The role that symmetries on elliptic curves play in the efficiency of their arithmetics has been~widely discussed e.g. in \cite{Koh11} and \cite{Koh12}. Kohel noticed in \cite{Koh12} that symmetries obtained by an automorphism group $\{ [\pm 1] \}$ and translation by the specific points of proper order have impact on the efficiency of addition law. He gave Hessian and Edwards curves as the most representative examples here.

In \cite{Koh12}, Kohel was studying symmetric quartic models over binary fields with a rational $4$-torsion point $T$. According to \cite{Koh12}, a genus 1 curve admits translations by rational points and translation morphism $\tau_{T}=P+T$ on curve $E$ is projectively linear (induced by a linear transformation of the ambient projective space), iff $E$ is a degree $n$ model determined by a complete linear system in $\mathbb P^{n-1}$ and $T$ is in the $n$-torsion subgroup.

In this paper, we use these ideas to identify new compression functions of high degree ($>2$) especially for Huff's curves, generalized Hessian and Hessian curves. The compression functions for which we are looking are invariant on the action of involution and translation by specific point $T$ of order $n$, meaning that for compression function of degree $f_{2n}(P)=f_{2n}(Q)$ holds iff $Q=\pm P +[k]T$, for $k=\ov{0,n-1}$.

In the case of Huff's curves, where $E_{Hu}/\K\ :\ ax(y^2-1) = by(x^2-1)$, we will study their arithmetics using a high-degree point compression. Compression function $f$ of order $4$ is obtained by using symmetry given by point $(a:b:0)$ of order $2$. A compression function of order $8$ is obtained by using symmetry given by three points of order $2$: $(a:b:0), (1:0:0)$ and $(0:1:0)$. Finally, a compression function of degree $16$ is obtained by using three $2$-torsion points $(a:b:0), (1:0:0)$ and $(0:1:0)$, as well as one point of order $4$ of the form $(\pm 1: \pm 1: 1)$. Let us note that compression functions of different degrees were obtained by Farashahi and Hosseini in \cite{Far17}, where also translations $\tau_T$ by a proper point of order $2$ and $4$ were used in the case of twisted Edwards curves.

In the case of generalized Hessian curves, given by $E_{GH}: x^3+y^3+a=dxy$, we will study the arithmetics of these curves using a point compression of degree $6$, with this compression function obtained by using symmetry given by a $3$-torsion point $(1:-\omega:0)$.

In the case of Hessian curves, given by $E_{GH}:x^3+y^3+1=dxy$, we will study the arithmetics of these curves using a point compression of degree $18$, with this compression function obtained by using symmetry given by two $3$-torsion points $(1:-\omega:0)$ and $(-\omega:0:1)$. More details about this approach will be presented in subsection \ref{GH18deg}.

In \cite{Dry19} a method for automating the process of searching for doubling and differential addition formulas for compression functions of order $2$ is presented. This method uses the Gr\"obner basis mechanism. Because in the case of high-degree compression functions manual investigations identyfying differential and doubling formulas are very difficult, we modified the ideas from \cite{Dry19} and implemented a Magma program, which may be used to search for differential addition and doubling formulas also in the case of compression functions of a degree higher than 2. The method of searching for convenient formulas may be very memory-consuming. It was necessary to use a computer with 384 GB of RAM to find such a formula in some cases.

Finally, we prove that if for a model $E$ of an elliptic curve exists an isomorphism $\phi:E \to E_M$, where $E_M$ is the Montgomery curve and for any $P \in E(\K)$ holds that $\phi(P)=(\phi_x(P), \phi_y(P))$, then for a model $E$ one may find a compression function of degree $2$ and, defined for this compression function, differential addition and doubling formulas, respectively, $A$ and $D$ of the same efficiency (in the whole paper, by the efficiency, we mean computational efficiency, which is the required number of elementary operations) as Montgomery's. However, such compression functions of degree $2$ may be sometimes complicated, as same as constants appearing in differential addition and doubling formulas and therefore may be not optimal for all applications.

Basing on this theorem we also affirm, that for a family of elliptic curves having natural point of order $3$, e.g. Hessian, twisted Hessian and generalized Hessian curves, obtaining formulas of the same or similar efficiency as for the Montgomery curve is impossible, because there do not exist natural isomorphisms from these curves to the Montgomery curve.\vspace*{-2mm}

\section{Compression functions}  \label{sec2}
In this section we provide basic facts on doubling, differential addition and point recovery. We also
describe a method based on the Gr\"obner bases from \cite{Dry19}, with modifications for searching for formulas concerning high-degree compression functions. This method was used, in Magma, to search for the formulas given in the following sections (for an introduction to the Gr\"obner bases theory, see \cite{Ada94} or \cite{Cox94}).

Peter  Montgomery \cite{Mon87} gave some efficient and simple formulas for point doubling and differential addition after compression on elliptic curves $By^2 = x^3 + Ax^2 +x$.  These formulas may be given for  any model of an elliptic curve.
Let $E$ be an elliptic curve over a field $\K$. In such a case a function $f:E\to \K$ for which holds that $f(P) = f(Q)$ iff $Q=\pm P$ for all $P\in E$ is called a degree 2 compression function. We have induced point multiplication of values $f(P)$ given by $[n]f(P) = f([n]P)$ for $n\in \Z$. There exist rational functions  doubling $D(x)\in \K(x)$ and differential additions  $A_1(x,y),A_2(x,y)\in \K(x,y)$ after compression such that
\begin{equation}\label{D} f([2]P) = D(f(P)),\vspace*{-2mm}\end{equation}
\begin{equation}\label{A1}  f(P+Q) + f(P-Q) = A_1(f(P),f(Q)),\vspace*{-1mm}\end{equation}
\begin{equation}\label{A2} f(P+Q)f(P-Q) = A_2(f(P),f(Q)).\end{equation}
These properties allow to compute, after compression, $[n]f(P)$ using the Montgomery ladder algorithm. We may adopt
$A(x,y,z) =A_1(x,y)) - z$ or  $A(x,y,z) = A_2(x,y)/z$ in this algorithm.\medskip

\begin{algorithm}[H] \label{ALG:1}\caption{The Montgomery ladder}
\KwIn{$f(P)$ and the binary expansion of $n=(n_k,\ldots, n_0)_2$}
\KwOut{$[n]f(P)$}
$x_1:=f(P)$; $x_2 :=[2]x_1$; \\[-1pt]
\For{$i=k-1,\ldots, 0$}
{
\eIf{$n_i=1$}{
$x_1:=A(x_1,x_2,f(P))$;\\[-1pt]
$x_2:=D(x_2)$;
}
{
$x_2:=A(x_1,x_2,f(P))$;\\[-1pt]
$x_1:=D(x_1)$;
}
}
\Return{$x_1$}\;
\end{algorithm}\medskip

Formulas for doubling and differential addition were given for standard models of elliptic curves: Montgomery, Weierstrass, Edwards, Hessian, Jacobi quartic, and Huff's curves.

One may also consider compressions of higher degrees. In general, the degree of compression function $g$ is the number of different elements $Q=\pm P+[k]T$, for $k=\overline{0,n-1}$ and $T$ being point of order $n$, for which equation $g(P)=g(Q)$ holds for every $P \in E(\K)\setminus S$. Set $S$ contains points of order $2$ and points of order $n$ from subgroup $\langle T \rangle$. In this case the degree of compression function $g$ is equal to $2n$. It is worth noting that the degree of compression function is always even.

For a function $g:E\to \K$, there exist rational functions $D(x)\in \K(x)$ and $A(x,y,z)\in \K(x,y,z)$ such that $g([2]P) = D(g(P))$ and $g(P+Q) = A(g(P),g(Q), g(P-Q))$ for generic points on $E$, then we have induced multiplication $[n]g(P) = g([n]P)$ for $n\in \N$
which may be computed using the Montgomery ladder algorithm (see also \cite[Sec.5.4]{Far10}). Note  that multiplication $[n]g(P) = g([n]P)$ is independent of choosing a point $P$ which may be checked by induction on $n$. Let $g(P) = g(P')$. For doubling, we have $g([2]P) = D(g(P)) = D(g(P')) = g([2]P')$. Let us assume that for each $0\leq k\leq n$ we have $g([k]P)= g([k]P')$, then we have $g([n+1]P) = A(g([n]P), g(P), g([n-1]P)) =A(g([n]P'), g(P'), g([n-1]P')) =g([n+1]P')$.

In section \ref{AlgGrob}, we remind, from \cite{Dry19}, a method (with some modifications) used to search for functions
$D,A_1,A_2$ for compressions $g$ of degrees $\geq 2$, which method was used in \cite{Dry19} for compressions of degree 2.

Compressions of higher degrees were given for Edwards \cite{Far17} and Jacobi quartic \cite{Gu12} curves.
Natural examples of low-order subgroups are known for Edwards, Hessian, Huff's, and Jacobi quartic elliptic curves.
Given a subgroup $G$ in the generic model of an elliptic curve, one may try to obtain compression $g$ of degree $2|G|$ such that
$g(\pm P + G) = g(P)$ for each $P\in E$.

Now will be presented approach to searching for compression functions of degree $>2$.

\subsection{Compression functions of high degree using symmetries on elliptic curves}
\label{HDegComp}
In this subsection a method for obtaining compression functions of high degree using natural symmetries on a given model of an elliptic curve will be presented.

At first, let us consider translation $\tau_{T}:E \to E, \tau_{T}(P)=P+T$ for a certain chosen point $T \in E(\K)$ of order $n$.
We will be searching for the compression function $f_{2n}$ of degree $2n$ which is invariant under involution and translation by $T$. This means that $f_{2n}(P)=f_{2n}\left( Q \right)$ iff $Q=\pm P +[k]T$, for $k=\ov{0,n-1}$.

\begin{prop}
\label{prop1}
Let us note, that such a function may be easily found for a certain model $E$ of an elliptic curve if three conditions hold:
\begin{itemize}
\item involution $[-1]P$ is projectively linear, which means that if $P=(X:Y:Z)$, then $[-1]P=(\alpha_1 X+\beta_1 Y+\gamma_1 Z: \alpha_2 X+\beta_2 Y+\gamma_2 Z: \alpha_3 X+\beta_3 Y+\gamma_3 Z)$ for some constant $\alpha_i,\beta_i, \gamma_i \in \K, i=\overline{1,3}$,
\item point $T$ of order $n$ naturally belongs to $E(\K)$,
\item translation $\tau_{T}:E \to E: \tau(P)=P+T$ is also projectively linear.
\end{itemize}
\end{prop}

This approach, using symmetries of involution and translation, will be used for obtaining efficient compression functions on Edwards, Huff's and Hessian family of elliptic curves in section \ref{cfaltellc}.

\begin{rem}
Using the approach presented in Proposition \ref{prop1}, the process of searching for a compression function of a high degree should consist of the following steps:
\begin{enumerate}
\item at first, use the point addition formula and find equations for $\tau_T=P+T$, where $T \in E(\K)$ is point of order $n$ and $P$ is any point in $E(\K)$,
\item check if equation for $\tau_T$ is projectively linear,
\item let us try to find a compression function of degree $2n$ using the character of $\tau_T$.
\end{enumerate}
\end{rem}

\begin{rem}
Let us know that if on an elliptic curve $E$ there is a point $T \in E(\K)$ of order $n$, then one can always construct compression function of degree $2n$ \cite{Fau14}. It is possible by constructing an isogeny $\psi:E \to E/\langle nT \rangle$. Then a compression function of degree $2n$ may be obtained using a compression function of degree $2$ and finally $f_{2n}(P)=f_2(\psi(P))$. Even though, compression functions of degree $2n$ constructed in this way may be not so efficient.
\end{rem}

\subsection{Algorithms to determine formulas used in the compression}
\label{AlgGrob}

In the case of function $g:E\to K$, formulas \eqref{D}, \eqref{A1}, \eqref{A2} may be searched for using the method from \cite{Dry19}, with some small modifications. Assume for simplicity that $E$ is contained in $\mathbb P^2$ and is given by the equation
$E:w(x,y) = 0$ in $K^2$ for $w(x,y)\in K[x,y]$.

\medskip
Let us assume that we indent to check if there exists a formula for doubling
$D \in K(x)$ satisfying  \begin{equation}\label{compD} D(g(x,y)) = g([2](x,y))\end{equation} on $E$, where
$D(x) =  \frac{D_1(x)}{D_2(x)}$,  $D_1, D_2\in K[x]$ are polynomials of degrees $d_1,d_2$ at most, respectively, for fixed bounds $d_i$.
Let  $D_1 = \sum_{\al\leq d_1} a_{\al} x^{\al}$ and $D_2 = \sum_{\be\leq d_2} b_{\be} x^{\be}$, with unknown coefficients
$a_{\al}, b_{\be}$. We may  write $D(g(x,y)) = \frac{v_1(x,y)}{v_2(x,y)}$, where  $v_1, v_2$ are polynomials in $x,y$, whose
coefficients contain $a_{\al}, b_{\be}$ of degree one.   Writing $g([2](x,y)) = \frac{u_1(x,y)}{u_2(x,y)}$, where $u_1,u_2\in K[x,y]$,
we intend to determine the values of $a_{\al}, b_{\be}$ such that $v_1u_2- v_2u_1 \in (w)$, which is equivalent to \eqref{compD}.
Since $v_1u_2- v_2u_1$ contains $a_{\al}, b_{\be}$ of degree one, the normal form $N(v_1u_2- v_2u_1)$ with respect to the ideal $(w)$
contains $a_{\al}, b_{\be}$ also of degree one. Hence, in order to determine $a_{\al}, b_{\be}$ for which $N(v_1u_2- v_2u_1) = 0$, we need to solve a system of linear equations when coefficients depending on $a_{\al},b_{\be}$  of the normal form are zero.

\medskip
Let us assume that we intend to determine a function $A_2(x,y)\in K(x,y)$ such that
\begin{equation}\label{comp3}  g((x_1,y_1) + (x_1,y_1))g((x_1,y_1) - (x_2,y_2)) = A_2(g(x_1,y_1),g(x_2,y_2))\end{equation}
on $E\times E$.  Let $w_i = w(x_i,y_i)$ for $i=1,2$.
Let $A_2 = \frac{u_1(x,y)}{u_2(x,y)}$, where  $u_1 = \sum_{|\al| \leq d_1} a_{\al}x^{\al_1}y^{\al_2}$,
$\al = (\al_1,\al_2)\in \N^2$, $|\al| = \al_1+\al_2$, and, similarly  $u_2 = \sum_{ |\be| \leq d_2} b_{\be}x^{\be_1}y^{\be_2}$
for given bounds concerning degrees $d_1,d_2\in \N$.

\medskip
Similarly as above, we may write $A_2(g(x_1,y_1),g(x_2,y_2)) = \frac{v_1}{v_2}$, where $v_1,v_2$ are polynomials in $x,y$, which contain unknown
coefficients $a_{\al}, b_{\be}$ of degree at most one. Writing $g((x_1,y_1) + (x_1,y_1))g((x_1,y_1) - (x_2,y_2)) = \frac{g_1}{g_2}$, where $g_1,g_2\in K[x,y]$ we intend to determine the values of $a_{\al},b_{\be}$ such that  $g_1v_2 - g_2v_1$ belongs to the ideal $I=(w_1,w_2)$, so the normal form $N(g_1v_2 - g_2v_1)$ with respect to $I$ is equal to 0. Similarly as above, this leads to the system of linear equations with respect
to $a_{\al}, b_{\be}$.

\section{Alternative models of elliptic curves}

In this section alternative models of elliptic curves will be briefly discussed.

\subsection{Edwards curves}
\begin{defin}
The Edwards curve $E_{Ed}$ over a field $\K$ is given by the equation \cite{Ber07}
  \begin{equation}\label{Ed}
    E_{Ed}/\K\ :\ x^2+y^2=1+dx^2y^2,
\end{equation}
where $d\not \in \{0,1\}$.
\end{defin}
The sum of points $P=(x_1,y_1)$ and $Q=(x_2,y_2)$ on $E_{Ed}$ is given by following formula:
\begin{equation}
P+Q=\left(\frac{x_1y_2+y_1x_2}{1+dx_1x_2y_1y_2},  \frac{y_1y_2-x_1x_2}{1-dx_1x_2y_1y_2}\right).
\end{equation}

The neutral element is $\mathcal{O}=(0,1)$ and the negation is given by $-(x,y)=(-x,y)$.
If $d$ is not a square in $\K$, then the addition formula presented above is complete in the set of $\K$-rational points on~$E$.

\subsection{Generalized and twisted Hessian curves}

In this section, basic definitions on generalized Hessian and twisted Hessian curves will be presented.
\begin{defin}
  The generalized Hessian curve $E_{GH}$ over a field $\K$ is given by the following equation~\cite{Far10}
  \begin{equation}
    E_{GH}/\K\ :\ x^3+y^3+a=dxy,
  \end{equation}
  for $a,d\in \K$ where $a\ne 0$ and $d^3\ne 27a$.
\end{defin}
The sum of points $P=(x_1,y_1)$ and $Q=(x_2,y_2)$ on $E_{GH}$ is given by the following unified formula, which works for all inputs of $P,Q \notin T_{\zeta}$, where $T_{\zeta}=\{(-\zeta:0:1)| \zeta\in \ov{\mathbb F}, \zeta^3=a\}$:
\begin{equation}
P+Q=\left( \frac{a y_1-x_2 y_2 x_1^2}{x_1 x_2^2-y_2 y_1^2}, \frac{x_1 y_1 y_2^2-a x_2}{x_1 x_2^2-y_2 y_1^2} \right).
\end{equation}
Alternatively, the sum of points $P=(x_1,y_1)$ and $Q=(x_2,y_2)$ on $E_{GH}$ is given by the following formulas:
\begin{itemize}
\item if $P\ne \pm Q$ (point addition)
\begin{equation}
P+Q=\left(\frac{y_1^2x_2-y_2^2x_1}{x_2y_2-x_1y_1}, \frac{x_1^2y_2-x_2^2y_1}{x_2y_2-x_1y_1}  \right),
\end{equation}
\item if $P=Q$ (point doubling)
\begin{equation}
  [2]P=\left(\frac{y_1(a-x_1^3)}{x_1^3-y_1^3}, \frac{x_1(y_1^3-a)}{x_1^3-y_1^3}  \right).
\end{equation}

\end{itemize}
The neutral element is a point at infinity $(1:-1:0)$. The negation of the point $P=(x_1,y_1)$ is $-P=(y_1,x_1)$.

\begin{defin}
The twisted Hessian curve $E_{TH}$ over a field $\K$ is given by the equation \cite{Ber15}
\begin{equation}
  E_{TH}/\K\ : \ov{a}\ov{x}^3+\ov{y}^3+1=\ov{d}\ \ov{x}\ \ov{y}
\end{equation}
for $\ov{a},\ov{d}\in \K$ where $\ov{a}\ne 0$ and $\ov{d}^3\ne 27\ov{a}$.
\end{defin}

The neutral element of addition law for twisted Hessian curves is the point $(0,-1)$. The negation of the point $\ov{P}=({\ov{x}}_1,{\ov{y}}_1)$ is $-\ov{P}=({\ov{x}}_1/{\ov{y}}_1,1/{\ov{y}}_1)$. The sum of points $\ov{P}=({\ov{x}}_1,{\ov{y}}_1)$ and $\ov{Q}=({\ov{x}}_2,{\ov{y}}_2)$ on $E_{TH}$ is given by the following formulas:
\begin{itemize}
\item where $\ov{P}\ne \pm \ov{Q}$ (point addition)
\begin{equation}
\ov{P}+\ov{Q}=\left(\frac{{\ov{x}}_1-{\ov{y}}_1^2{\ov{x}}_2{\ov{y}}_2}{\ov{a}{\ov{x}}_1{\ov{y}}_1{\ov{x}}_2^2-{\ov{y}}_2}, \frac{{\ov{y}}_1{\ov{y}}_2^2-\ov{a}{\ov{x}}_1^2{\ov{x}}_2}{\ov{a}{\ov{x}}_1{\ov{y}}_1{\ov{x}}_2^2-{\ov{y}}_2}  \right)
\end{equation}
\item where $\ov{P}=\ov{Q}$ (point doubling)
\begin{equation}
  [2]\ov{P}=\left(\frac{{\ov{x}}_1-{\ov{y}}_1^3{\ov{x}}_1}{\ov{a}{\ov{y}}_1{\ov{x}}_1^3-{\ov{y}}_1}, \frac{{\ov{y}}_1^3-\ov{a}{\ov{x}}_1^3}{\ov{a}{\ov{y}}_1{\ov{x}}_1^3-{\ov{y}}_1}  \right).
\end{equation}
\end{itemize}

Although the model of twisted Hessian curves seems to be used more frequently, we chose the generalized Hessian curves model. There are two reasons behind such a decision. First of all, there is birationally equivalence between twisted Hessian and generalized Hessian models.

\begin{rem}
\label{rem3}
In projective coordinates the generalized Hessian curve is given by the equation
\begin{equation}
  E_{GH}/\K\ :\ X^3+Y^3+aZ^3=dXYZ.
\end{equation}
By swapping $X$ with $Z$, we get the equation of a twisted Hessian curve in projective coordinates
\begin{equation}
  E_{TH}: \ov{a}\ov{X}^3+\ov{Y}^3+\ov{Z}^3=\ov{d}\ \ov{X}\ \ov{Y}\ \ov{Z},
\end{equation}
The isomorphism $\phi$ between $E_{GH}$ and $E_{TH}$, is given by $\phi:E_{GH} \to E_{TH}$, $\phi(X:Y:Z)=(Z:Y:X)$, $\ov{a}=a$, $\ov{d}=d$.
\end{rem}

The other reason is that the generalized Hessian curve (in affine model) is symmetrical and the twisted Hessian curve is not. Relying on this fact, it was easier to construct a compression function acting on 3-torsion points in this case.

\subsection{Huff's curves}

\begin{defin}
  The Huff's curve $E_{Hu}$ over a field $\K$ is given by the equation \cite{Joy10}
  \begin{equation}\label{Hu}
    E_{Hu}/\K\ :\ ax(y^2-1) = by(x^2-1),
\end{equation}
where $a^2\neq b^2$ and $a,b\neq 0$.
\end{defin}
The sum of points $P=(x_1,y_1)$ and $Q=(x_2,y_2)$ on $E_{Hu}$ is given by the following complete formula:
\begin{equation}
P+Q=\left(\frac{(x_1+x_2)(1+y_1y_2)}{(1+x_1x_2)(1-y_1y_2)}, \frac{(y_1+y_2)(1+x_1x_2)}{(1-x_1x_2)(1+y_1y_2)}  \right),
\end{equation}
Alternatively, to compute $P+Q$, one may use following formulas:
\begin{itemize}
\item if $P\ne \pm Q$ (point addition)
\begin{equation}
P+Q=\left(\frac{(x_1-x_2)(y_1+y_2)}{(1-x_1x_2)(y_1-y_2)}, \frac{(y_1-y_2)(x_1+x_2)}{(1-y_1y_2)(x_1-x_2)}  \right),
\end{equation}
\item if $P=Q$ (point doubling)
\begin{equation}
  [2]P=\left(\frac{(2y_1^2+2)x_1}{(x_1^2+1)y_1^2-x_1^2-1},\frac{(2x_1^2+2)y_1}{(x_1^2-1)y_1^2+x_1^2-1}\right).
\end{equation}
\end{itemize}
Point  $O=(0,0)$ is the neutral element, and  the opposite point $-(x,y)=(-x,-y)$. In projective coordinates the equation, (\ref{Hu}) has a form:
\begin{equation}\label{HuProj}
E_{Hu}/\K\ :\ aX(Y^2-Z^2) = bY(X^2-Z).
\end{equation}
There are three points at infinity on $E_{Hu}$: $T_1 = (1:0:0)$, $T_2 = (0:1:0)$ and $T_3 = (a:b:0)$. Points $T_1$, $T_2$ and $T_3$ are of order 2. Additionally there are four points of order 4: $(1:1:1)$, $(-1:1:1)$, $(1:-1:1)$ and $(-1:-1:1)$ (e.g. $(1,1)$, $(-1,1)$, $(1,-1)$ and $(-1,-1)$ in the affine space).

\section{High-degree compression function on alternative models of elliptic curves}
\label{cfaltellc}

In this section we will present compression functions of degree $\geq 2$ for Edwards, Huff's and Hessian family of elliptic curves. However, it is worth noteing that a compression function of order $4$ for Jacobi quartics has been also proposed in \cite{Gu12}. We will mainly focus on compression functions that are new for these models and have not been presented before.

\subsection{Edwards curves}

\subsubsection{Compression function of degree $2$}
Edwards curves were widely analyzed in the context of their arithmetics using compression functions. The obvious compression function of degree $2$ is $f_2(x,y)=y$. The arithmetic using this compresison function may be found, for example, in \cite{Cas08}, but we recall this arithmetic below. If $f_{2}(P)=r_P$ and $f_{2}(Q)=r_Q$, then the differential addition $f_2(P+Q)f_2(P-Q)$ is given by the following formula:
\begin{equation}
\label{dadd2ED}
\small
\begin{array}{rl}
f_2(P+Q) f_2(P-Q)=-\frac{(d r_P^2-1) r_Q^2-r_P^2+1}{(d r_P^2-d) r_Q^2-d r_P^2+1}.
\end{array}
\normalsize
\end{equation}

The formula for doubling has the following form:
\begin{equation}
\label{doub2ED}
f_2([2]P)=-\frac{d r_P^4-2 r_P^2+1}{d r_P^4-2 d r_P^2+1}.
\end{equation}

\begin{expl}
Formula (\ref{dadd2ED}) may be obtained using the algorithm from Appendix \ref{appDiffMain}, with modifications from Appendix \ref{appDiff2ED}. Accordingly, formula (\ref{doub2ED}) may be obtained using the algorithm from Appendix \ref{appDoubMain}, with modifications from Appendix \ref{appDoub2ED}.
\end{expl}

\subsubsection{Compression function of degree $4$}
We will give an example of a compression function of degree $4$ for the Edwards curve $E_{Ed}: x^2+y^2=1+dx^2y^2$.
A compression function of degree $4$ may be given by $f_4(x,y)=y^2$. It is easy to show, that $f_4$ has the degree of $4$. At first, let us note, that
\begin{itemize}
\item involution $[-1]P$ is projectively linear, because $[-1]P=(-X:Y:Z)$ for $P=(X:Y:Z)$,
\item point $T=(0:-1:1)$ of order $2$ naturally belongs to $E_{Ed}(\K)$,
\item translation $\tau_{T}:E_{Ed} \to E_{Ed}: \tau(P)=P+T$ is also projectively linear, because if $P=(X:Y:Z) \in E_{Ed}(\K)$, then $P+(0:-1:1)=(-X: -Y: Z)$.
\end{itemize}
Let us note that $f_{4}(P)=f_4(Q)$, iff $Q=\pm P + [k](0:-1:1)$, for $k=\ov{0,1}$ is in set $S=\{(x,y),(-x,y),(x,-y),(-x,-y)\}$, for $P=(x,y)$.

\medskip
Let us assume that $r=y^2$ for $y \neq 0$. Using this identity in the Edwards curve equation, one may obtain that
\begin{equation}
x^2+r=1+dx^2r.
\end{equation}
This means that
\begin{equation}
\label{EdwDeg4}
x^2(rd-1)-r=0.
\end{equation}

Because equation (\ref{EdwDeg4}) is a polynomial of order $2$, it means that one may find two roots of such a polynomial at most. This means, that for every $r$ one has at most 2 distinct values of $x$. Because on an Edwards curve there are always exactly two points having the same $x$-coordinate, it means that equation (\ref{EdwDeg4}) may be satisfied by $4$ points at most. All of these points belong to set $S$, which may be easily checked manually.

\begin{theorem}
A differential addition formula for $f_{4}(P+Q)f_{4}(P-Q)$ on an Edwards curve and a compression function $f_{4}(x,y)=y^2$, where $f_{4}(P)=r_P$ and $f_{4}(Q)=r_Q$, is given by:
\begin{equation}
\label{dadd4ED}
\small
\begin{array}{rl}
f_4(P+Q) f_4(P-Q)=\frac{L}{M},
\end{array}
\normalsize
\end{equation}
where
\begin{equation}
\scriptsize
\begin{array}{rl}
L=(d^2{r_P}^2-2ad{r_P}+a^2){r_Q}^2+(-2ad{r_P}^2+(2ad+2a^2){r_P}-2a^2){r_Q}+a^2{r_P}^2-2a^2{r_P}+a^2,\\
M=(d^2{r_P}^2-2d^2{r_P}+d^2){r_Q}^2+(-2d^2{r_P}^2+(2d^2+2ad){r_P}-2ad){r_Q}+d^2{r_P}^2-2ad{r_P}+a^2.
\end{array}
\normalsize
\end{equation}
for every $P \in E_{Ed}(\K)$ holds $f(P)=r_P$.

\medskip
Similarly, doubling is given by
\begin{equation}
\label{doub4ED}
f_4([2]P)=\frac{d^2{r_P}^4-4ad{r_P}^3+(2ad+4a^2){r_P}^2-4a^2{r_P}+a^2}{d^2{r_P}^4-4d^2{r_P}^3+(4d^2+2ad){r_P}^2-4ad{r_P}+a^2}.
\end{equation}
\end{theorem}

\begin{expl}
Formula (\ref{dadd4ED}) may be obtained using the algorithm from Appendix \ref{appDiffMain}, with modifications from Appendix \ref{appDiff4ED}. Accordingly, formula (\ref{doub4ED}) may be obtained using the algorithm from Appendix \ref{appDoubMain}, with modifications from Appendix \ref{appDoub4ED}.
\end{expl}

\subsubsection{Compression function of degree $8$}

A compression function of degree $8$ on an Edwards curve, was presented by Farashahi and Hosseini in \cite{Far17}. They gave the example of a compression function $d x^2 y^2$ of degree 8 for the Edwards curve $E_{Ed}: x^2+y^2=1+dx^2y^2$. Below, we present results for a similar compression function of degree $8$ on curve $E_{Ed}$ given by $f_8(x,y)=x^2 y^2$. It is easy to show, that $f_8$ has the degree of $8$. At first, let us note, that
\begin{itemize}
\item involution $[-1]P$ is projectively linear, because $[-1]P=(-X:Y:Z)$ for $P=(X:Y:Z)$,
\item point $T=(1:0:1)$ of order $4$ naturally belongs to $E_{Ed}(\K)$,
\item translation $\tau_{T}:E_{Ed} \to E_{Ed}: \tau(P)=P+T$ is also projectively linear, because if $P=(X:Y:Z) \in E_{Ed}(\K)$, then $P+(1:0:1)=(Y: -X: Z)$.
\end{itemize}
Let us note that $f_{8}(P)=f_8(Q)$, iff $Q=\pm P + [k](1:0:1)$, for $k=\ov{0,3}$ is in set $S=\{(x,y),(-x,y),(x,-y),(-x,-y), (y,x),(-y,x),(y,-x),(-y,-x)\}$, for $P=(x,y)$.

\medskip
Let us assume that $r=x^2 y^2$ for $x,y \neq 0$. Using this identity in the Edwards curve equation, one may obtain that
\begin{equation}
x^2+\frac{r}{x^2}=1+dr.
\end{equation}
This means that
\begin{equation}
\label{EdwDeg8}
x^4-(dr+1)x^2+r=0.
\end{equation}

Because equation (\ref{EdwDeg8}) is a polynomial of order $4$, it means that one may find four roots of such a polynomial at most. This means, that for every $r$ one has at most 4 distinct values of $x$. Because on an Edwards curve there are always at most two points having the same $x$-coordinate, it means that equation (\ref{EdwDeg4}) may be satisfied by $4$ points at most. All of these points belong to set $S$, which may be easily checked manually.

\begin{theorem}
A differential addition formula for $f_{8}(P+Q)f_{8}(P-Q)$ on an Edwards curve and a compression function $f_{8}(x,y)=x^2 y^2$, where $f_{8}(P)=r_P$ and $f_{8}(Q)=r_Q$, is given by:
\begin{equation}
\label{dadd8ED}
\small
\begin{array}{rl}
f_8(P+Q) f_8(P-Q)=\frac{(r_P-r_Q)^2}{(d^2r_Pr_Q-1)^2}.
\end{array}
\normalsize
\end{equation}
Similarly, doubling is given by
\begin{equation}
\label{doub8ED}
f_4([2]P)=\frac{4 d^2 r_P^3+(8 d-16a) r_P^2+4x}{d^4 r_P^4-2 d^2 r_P^2+1}.
\end{equation}
\end{theorem}

\begin{expl}
Formula (\ref{dadd8ED}) may be obtained using the algorithm from Appendix \ref{appDiffMain}, with modifications from Appendix \ref{appDiff8ED}. Accordingly, formula (\ref{doub8ED}) may be obtained using the algorithm from Appendix \ref{appDoubMain}, with modifications from Appendix \ref{appDoub8ED}.
\end{expl}

\subsection{Hessian, generalized Hessian and twisted Hessian curves}

\subsubsection{Compression function of degree 2}

Let us define a compression function on a generalized Hessian curve of degree 2 given by $f_2(P)=x+y$. This function may be obtained from the function $\ov{f}_2(\ov{P})=\frac{\ov{y}+1}{\ov{x}}$ from \cite{Dry19}, using isomorphism between $E_{GH}$ and $E_{TH}$. Using the same isomorphism between $f_2(P)$ and $\ov{f}_2(\ov{P})$, one may use differential addition and doubling formulas from \cite{Dry19} and obtain that if $r_P=f_2(P)$ and $r_Q=f_2(Q)$ then $f(P+Q)f(P-Q)$ may be presented in the following form:
\small
\begin{equation}
\label{dadd2MultGH}
f_2(P+Q)f_2(P-Q)=\frac{(d{r_P}^2-3a){r_Q}^2+(6a{r_P}+2ad){r_Q}-3a{r_P}^2+2ad{r_P}+ad^2}{(3{r_P}+d){r_Q}^2+(3{r_P}^2+d{r_P}){r_Q}+d{r_P}^2-3a)}.
\end{equation}
\normalsize
In the same manner, $f(P+Q)+f(P-Q)$ may be presented as
\small
\begin{equation}
\label{dadd2AddGH}
f_2(P+Q)+f_2(P-Q)=-\frac{((3{r_P}^2+d{r_P}){r_Q}^2+(d{r_P}^2+d^2{r_P}+6a){r_Q}+6a{r_P}+2ad}{(3{r_P}+d){r_Q}^2+(3{r_P}^2+d{r_P}){r_Q}+d{r_P}^2-3a}.
\end{equation}
\normalsize

Similarly, doubling may be presented as:
\begin{equation}
\label{doub2GH}
f_2([2]P)=\frac{-(r_P^4+4ar_P+ad)}{(2r_P^3+dr_P^2-a)}.
\end{equation}

\begin{expl}
Formula (\ref{dadd2MultGH}) may be obtained using the algorithm from Appendix \ref{appDiffMain},  with modifications from Appendix \ref{appDiff2MultGH}. Formula (\ref{dadd2AddGH}) may be obtained using the algorithm from Appendix \ref{appDiffMain},  with modifications from Appendix \ref{appDiff2AddGH} Accordingly, formula (\ref{doub2GH}) may be obtained using the algorithm from Appendix \ref{appDoubMain}, with modifications from Appendix \ref{appDoub2GH}.
\end{expl}

\subsubsection{Compression function of degree 6}

We will give an example of a compression function of degree $6$ for the generalized Hessian curve $E_{GH}: x^3+y^3+a=dx y$.
A compression function of degree $6$ may be given by $f_6(x,y)=xy$.

\medskip
It is easy to show, that $f_6$ has the degree of $6$. At first, let us note, that compression function $f_6$ fulfills all the criteria from Proposition \ref{prop1}:
\begin{itemize}
\item involution $[-1]P$ is projectively linear, because $[-1]P=(Y:X:Z)$ for $P=(X:Y:Z)$,
\item for every generalized Hessian curve over field $\K$, there exists root $\omega$ of polynomial $\omega^2+\omega+1=0$ in field $\ov{\K}$ and point $T\in E_{GH}\left( \ov{\K} \right)$ of order $3$ $T=(1: -\omega: 0)$ in projective coordinates,
\item translation $\tau_{T}:E_{GH} \to E_{GH}: \tau(P)=P+T$ is projectively linear, because if $P=(X:Y:Z) \in E(\K)$, then $P+(1: -\omega: 0)=(\omega X: \omega^{-1}Y:Z)$.
\end{itemize}
Let us note that $f_{6}(P)=f_6(Q)$, iff $Q=\pm P + [k](1:-\omega:0)$ for $k=\ov{0,2}$ is in set $S=\{(x,y),(y,x),(\omega x,\omega^2 y), (\omega^2 x,\omega y), (\omega y,\omega^2 x), (\omega^2 y,\omega x)\}$, for $P=(x,y)$.

\medskip
Let us assume that $r=xy$ for $x,y \neq 0$. Then $y=\frac{r}{x}$ and because $x^3+y^3+a=dxy$, then
\begin{equation}
x^3+\left( \frac{r}{x} \right)^3+a=dx\frac{r}{x}
\end{equation}
and
\begin{equation}
\label{deg6}
g(x)=x^6+(a-d r) x^3+r^3=0.
\end{equation}

Equation (\ref{deg6}) has $6$ roots at most. It is easy to show, that all points for which equation (\ref{deg6}) is satisfied belong to the set $S$ which is easy to check manually.

\begin{theorem}
\label{thmdeg6}
The differential addition formula for $f_{6}(P+Q)f_{6}(P-Q)$ on a generalized Hessian curve and compression function $f_{6}(x,y)=xy$, where $f_{6}(P)=r_P$ and $f_{6}(Q)=r_Q$, is given by:
\begin{equation}
\label{dadd6}
f_6(P+Q) f_6(P-Q)=\frac{\rp^2 \rqq^2-a d \rp \rqq+a^2 \rqq+a^2 \rp}{(\rqq-\rp)^2}.
\end{equation}

Similarly, doubling is given by
\begin{equation}
\label{doub6}
f_6([2]P)=\frac{\rp (a(d \rp-a)-\rp^3-a^2)}{(d \rp-a)^2-4 \rp^3}.
\end{equation}
\end{theorem}

\begin{expl}
Formula (\ref{dadd6}) may be obtained using the algorithm from Appendix \ref{appDiffMain}, without any modifications. Correspondingly, formula (\ref{doub6}) may be obtained using the algorithm from Appendix \ref{appDoubMain}, without any modifications.
\end{expl}

\begin{rem}
By comparing the formulas for $f_2(P+Q) f_2(P-Q)$ and $f_6(P+Q) f_6(P-Q)$, it is easy to notice that in the case of the differential addition function, $f_6(P)$ is more efficient. However, in the case of doubling, it seems that $f_2(P)$ has a lower computational cost than $f_6(P)$.
\end{rem}

\begin{rem}
Let us note, that Farashahi and Joye in \cite{Far10} obtained compression functions $f_6(x,y)=xy$ and $g_6(x,y)=x^3+y^3$ for binary generalized Hessian curves. Indeed, the same compression functions work also on generalized Hessian curves over fields with large characteristics. Let us see, that $g_6(x,y)=x^3+y^3=dxy-a=d \cdot f_6(x,y) - a$.
\end{rem}

\subsection{Compression function of degree 18}
\label{GH18deg}

In the previous subsections, we defined a compression function of degree $6$ on a generalized Hessian curve of degree $6$. In this subsection, we will be investigating a compression function of degree $18$ on a Hessian curve, using additional symmetries. Let us begin by noteing that for a Hessian curve given in projective coordinates
\begin{equation}
E_H: X^3+Y^3+Z^3=dXYZ
\end{equation}
if $P_1=(X:Y:Z) \in E_H(\K)$, then also $P_2=(X:Z:Y), P_3=(Y:X:Z), P_4=\linebreak (Y: Z:X), \, P_5=(Z:X:Y), \, P_6=(Z:Y:X) \in E_H(\K).\,$ Furthermore, let$\,$ us also$\,$ note, that if
\eject \noindent
$T_1=(1:-\omega:0)$ and $T_2=(-\omega:0:1)$, then:
\begin{enumerate}
\item translation $\tau_{T_1}:E_{GH} \to E_{GH}: \tau(P)=P+T$ is projectively linear, because $P+T_1=(\omega X: \omega^2 Y:Z) \in E_H(\K)$;
\item translation $\tau_{T_2}:E_{GH} \to E_{GH}: \tau(P)=P+T$ is projectively linear, because $P+T_2=(\omega Y: \omega^2 Z:X) \in E_H(\K)$.
\end{enumerate}

The above means that we will be searching for the compression function $f_{18}$ for which $f_{18}(P)=f_{18}(Q)$, iff $Q=\pm P+[k](1:-\omega:0)+[l](-\omega:0:1)$.

Now, we will give the following theorem.

\begin{theorem}
Let us state that $f_6(P)=xy=\frac{XY}{Z^2}$ is a compression function on a generalized Hessian curve and, therefore, on a Hessian curve. Because the Hessian curve equation is invariant under permutation of its coordinates $E(X:Y:Z) = E(Y:X:Z) = E(Z:Y:X) = E(Y:Z:X) = E(X:Z:Y) = E(Z:X:Y)$, then, using these symmetries, the compression function of degree 18 may be given by $f_{18}(P)=\frac{XY}{Z^2}+\frac{YZ}{X^2}+\frac{ZX}{Y^2}=\frac{x^3y^3+x^3+y^3}{x^2y^2}$.
\end{theorem}

\begin{proof}
At first, let us assume that $f_{18}(P)=R$, then
\begin{equation}
x^3+y^3=Rx^2 y^2-x^3y^3
\end{equation}
and substituting $x^3+y^3$ in the equation of the Hessian curve, one obtains that
\begin{equation}
Rx^2y^2-x^3y^3+1=dxy.
\end{equation}

\medskip
Let us note, that $xy=f_6(P)$. Let $r=f_6(P)$. Then
\begin{equation}
g(r)=-r^3+Rr^2-dr+1=0.
\end{equation}

For any $R$ there are at most three distinct roots of the polynomial $g$. Let us note that we showed that there are at most six distinct points in $E_H(\K)$ for which $f_6(P)=r$.

This means, that $f_{18}(P)$ has at most 18 distinct solutions. We will list all of those solutions in the set $S$, for $P=(x,y)$. It is easy to check that for every $Q \in S$ holds that $f_{18}(Q)=R$.

\medskip
The set $S$ is given by
\begin{equation}
\begin{array}{lr}
S=\Big\{ (X:Y:Z), (X:Z:Y), (Y:X:Z), (Y:Z:X), (Z:X:Y), (Z:Y:X),\\
(\omega X: \omega^2 Y:Z), (\omega X:\omega^2 Z:Y), (\omega Y:\omega^2 X:Z), (\omega Y:\omega^2 Z:X), (\omega Z:\omega^2 X:Y),\\
(\omega Z:\omega^2 Y:X), (\omega^2 X:\omega Y:Z), (\omega^2 X:\omega Z:Y), (\omega^2 Y: \omega X:Z), (\omega^2 Y:\omega Z:X),\\
(\omega^2 Z:\omega X:Y), (\omega^2 Z:\omega Y:X)  \Big\}
\end{array}
\end{equation}
in projective coordinates, which is equivalent to
\begin{equation}
\begin{array}{lr}
S=\Big\{ \left( x, y \right), \left( \frac{x}{y}, \frac{1}{y} \right), \left( y, x \right), \left( \frac{y}{x}, \frac{1}{x} \right), \left( \frac{1}{y}, \frac{x}{y} \right),\left( \frac{1}{x}, \frac{y}{x} \right),\\
\left( \omega x, \omega^2 y \right), \left( \omega \frac{x}{y}, \omega^2 \frac{1}{y} \right), \left( \omega y, \omega^2 x \right), \left( \omega \frac{y}{x}, \omega^2 \frac{1}{x} \right), \left( \omega \frac{1}{y}, \omega^2 \frac{x}{y} \right), \left( \omega \frac{1}{x}, \omega^2 \frac{y}{x} \right),\left( \omega^2 x, \omega y \right),\\
\left( \omega^2 \frac{x}{y}, \omega \frac{1}{y} \right), \left( \omega^2 y, \omega x \right), \left( \omega^2 \frac{y}{x}, \omega \frac{1}{x} \right), \left( \omega^2 \frac{1}{y}, \omega \frac{x}{y} \right), \left( \omega^2 \frac{1}{x}, \omega \frac{y}{x} \right) \Big\}.
\end{array}
\end{equation}
in affine coordinates.
\end{proof}

\begin{theorem}

The differential addition formula for $f_{18}(P+Q)f_{18}(P-Q)$ on a Hessian curve and a compression function $f_{18}(x,y)=\frac{x^3y^3+x^3+y^3}{x^2y^2}$, where $f_{18}(P)=r_P$ and $f_{18}(Q)=r_Q$, is given by:

\small
\begin{equation}
\label{dadd18}
\begin{array}{rl}
f_{18}(P+Q) f_{18}(P-Q)=\frac{{r_P}^2{r_Q}^2 + 9d{r_P}{r_Q} + (-4d^3 - 27){r_P} + (-4d^3 - 27){r_Q} + d^5 + 27d^2}{({r_P}-{r_Q})^2}.
\end{array}
\end{equation}
\normalsize
Similarly, doubling is given by
\small
\begin{equation}
\label{doub18}
\begin{array}{rl}
f_{18}([2]P)=\frac{\frac{1}{4}{r_P}^4 + \frac{9}{4}d{r_P}^2 + (-2d^3 - \frac{27}{2}){r_P} + \frac{1}{4} d^5 + \frac{27}{4}d^2}{{r_P}^3 - \frac{1}{4}d^2{r_P}^2 - \frac{9}{2}d{r_P} + d^3 + \frac{27}{4}}.
\end{array}
\end{equation}
\end{theorem}
\normalsize

\begin{expl}
Formula (\ref{dadd18}) may be obtained using the algorithm from Appendix \ref{appDiffMain}, with modifications from Appendix \ref{appDiff18GH}. Correspondingly, formula (\ref{doub18}) may be obtained using the algorithm from Appendix \ref{appDoubMain}, with modifications from Appendix \ref{appDoub18GH}.
\end{expl}

\subsection{Huff's curves}

\subsubsection{Compression function of degree 2}

Let us define the compression function on a Huff's curve of degree 2 given by $f_2(P)=xy$. This function was presented in \cite{Dry20}.
If $r_P=f_2(P)$ and $r_Q=f_2(Q)$, then the differential addition $f_2(P+Q)f_2(P-Q)$ is given by the following formula:
\begin{equation}
\label{dadd2HU}
f_2(P+Q)f_2(P-Q)=\left(\frac{r_P - r_Q}{r_Pr_Q - 1}\right)^2.
\end{equation}
The formula for doubling has the following form:
\begin{equation}
\label{doub2HU}
f_2([2]P)=\frac{4r_P (r_P^2+\left(\frac{a^2+b^2}{ab}\right)r_P+1}{(r_P^2-1)^2}.
\end{equation}

\begin{expl}
Formula (\ref{dadd2HU}) may be obtained using the algorithm from Appendix \ref{appDiffMain}, with modifications from Appendix \ref{appDiff2HU}. Correspondingly, formula (\ref{doub2HU}) may be obtained using the algorithm from Appendix \ref{appDoubMain}, with modifications from Appendix \ref{appDoub2HU}.
\end{expl}

\subsubsection{Compression function of degree 4}

In this subsection, a method for obtaining a compression function $f_4$ of degree $4$ using natural symmetries on Huff's curves and action on $2$-torsion point is presented.

\medskip
Let $T_3=(a:b:0) \in E_{Hu}(\K)$ be a point of order $2$  on  Huff's curve $E_{Hu}$ given by the equation (\ref{HuProj}). For a finite point $P=(X:Y:Z)=(x,y)\ne(0,0)$ the following translation:
\begin{equation}
\tau_{T_3}\ :\ P+T_3 = \left(\frac{1}{x},\frac{1}{y}\right)
\end{equation}
is projectively linear.
\begin{proof}\label{remHu4}
In order to verify, if the translation $\tau_{T_3}$ is linear in a projective space $\mathbb{P}^3$, we start by embedding the Huff's curve equation $E_{Hu}$ and the point $P$ into a $\mathbb{P}^1 \times \mathbb{P}^1$. We get the following Huff's curve equation:
\begin{equation}\label{HuPP}
E_{Hu}\ :\ aXZ_1(Y^2-Z_2^2)=aYZ_2(X^2-Z_2^1).
\end{equation}
In the space $(\mathbb{P}^1)^2$ for $P=((X:Z_1),(Y:Z_2))$ the translation $\tau_3$  has a form
\begin{equation}
  \tau_{T_3}\ :\ P+T_3 = ((Z_1:X),(Z_2:Y)).
  \end{equation}
By embedding the above solution into a projective space via Segre embedding $\rho  :
 \mathbb{P}^1\times \mathbb{P}^1\to \mathbb{P}^3$ given by
 \begin{equation}\label{Segre}
((X_1:X_2),(Y_1:Y_2))\to (X_1Y_1:X_1Y_2:X_2Y_1:X_2Y_2),
 \end{equation}
 we finally get
 \begin{equation}
   \begin{array}{l}
  P = (XY:XZ_2:YZ_1:Z_1Z_2)=(U_1:U_2:U_3:U_4),\\
  \tau_{T_3}\ :\ P+T_3 = (Z_1Z_2:YZ_1:XZ_2:XY) = (U_4:U_3:U_2:U_1).
\end{array}
  \end{equation}
  Additionally we have
\begin{equation}
 -P = (XY:-XZ_2:-YZ_1:Z_1Z_2)=(U_1:-U_2:-U_3:U_4).
\end{equation}

  In consequence, we see that the translation $\tau_3$  and the involution $[-1]P$ are projectively linear in $\mathbb{P}^3$.
\end{proof}

We will give an example of a compression function of degree $4$ for a Huff's curve $E_{Hu}: ax(y^2-1)=by(x^2-1)$.
A compression function of degree $4$ may be given by $f_4(x,y)=xy+\frac{1}{xy}$. It is easy to show, that $f_4$ has the degree of $4$. At first, let us note, that
\begin{itemize}
\item involution $[-1]P$ is projectively linear in $\mathbb{P}^3$ (see Remark \ref{remHu4}),
\item point $T_3=(a:b:0)$ of order $2$ naturally belongs to $E_{Hu}(\K)$,
\item translation $\tau_{T_3}:E_{Hu} \to E_{Hu}: \tau_{T_3}(P)=P+T_3$ is also projectively linear (see Remark \ref{remHu4}).
\end{itemize}
Let us note that $f_{4}(P)=f_4(Q)$, iff $Q=\pm P + [k](a:b:0)$, for $k=\ov{0,1}$ and $P=(x,y)$ is in set $S=\{(x,y), (-x,-y), (\frac{1}{x},\frac{1}{y}), (-\frac{1}{x},\frac{1}{y})\}$.

\medskip
Let us assume that $xy=t$ and $t\ne 0$. Then $f_4(x,y)=xy+\frac{1}{xy}=t+\frac{1}{t}$. One can denote $f_4(x,y)=r$. After short calculations we get
\begin{equation}
\label{Hdeg4}
h(t) = t^2-rt+1=0.
\end{equation}
Polynomial $h(t)$ has at most $2$ distinct roots. Let $t_1$ be a root of $h(t)$.
If we substitute $y=\frac{t_1}{x}$ in the Huff's curve equation (\ref{Hu}), we get
\begin{equation}\label{Huxy}
g(x)=(bt_1+a)x^2-bt_1-at_1^2=0.
\end{equation}
Equation (\ref{Huxy}) is quadratic and has at most two distinct roots. In consequence, the degree of the compression function $f_4$ is 4 at most.

\begin{theorem}
A differential addition formula for $f_{4}(P+Q)f_{4}(P-Q)$ on a Huff's curve and compression function $f_4(x,y)=xy+\frac{1}{xy}$, where $f_4(P)=r_P$ and $f_4(Q)=r_Q$, is given by:

\small
\begin{equation}
\label{daddhu4}
\begin{array}{rl}
f_{4}(P+Q)f_{4}(P-Q)=\frac{(ab)^2(r_Pr_Q + 4)^2 + 16(a^2 + b^2)(ab)(r_P + r_Q) + 16(a^2 + b^2)^2}{(ab)^2(r_P - r_Q)^2}.\\
\end{array}
\end{equation}\normalsize
Denoting $(a^2+b^2)/(ab)=A$, we get
\begin{equation}
\label{daddhu4b}
\begin{array}{rl}
f_{4}(P+Q)f_{4}(P-Q)=\frac{(r_Pr_Q + 4)^2 + 16A(r_P + r_Q) + 16A^2}{(r_P - r_Q)^2}.\\
\end{array}
\end{equation}
\normalsize
Similarly, doubling is given by
\begin{equation}
\label{doubhu4}
\small
\begin{array}{rl}
f_{4}([2]P)=\frac{(ab)^2(r_P^2 + 4)^2 + 32(a^2 + b^2)(ab)r_P + 16(a^2 + b^2)^2}{4ab((ab)(r_P^3 -4r_P) +(a^2 + b^2)r_P^2   -4(a^2 +b^2))}.\\
\end{array}
\normalsize
\end{equation}
Denoting $(a^2+b^2)/(ab)=A$ we get
\begin{equation}
\label{doubhu4b}
\begin{array}{rl}
f_{4}([2]P)=\frac{((r_P^2 + 4)+4A)^2 + 8A(r_P^2+4)}{4(r_P+A)(r_P^2 -4)}.
\end{array}
\end{equation}

\end{theorem}

\begin{expl}
Formula (\ref{daddhu4}) may be obtained using the algorithm from Appendix \ref{appDiffMain}, with modifications from Appendix \ref{appDiff4HU}. Correspondingly, formula (\ref{doubhu4}) may be obtained using the algorithm from Appendix \ref{appDoubMain}, with modifications from Appendix \ref{appDoub4HU}.
\end{expl}

\subsubsection{Compression function of degree 8}

In this subsection a method for obtaining a compression function $f_8$ of degree $8$ using natural symmetries on Huff's curves and action on three $2$-torsion points will be presented.

\medskip
Let $T_1, T_2,T_3 \in E_{Hu}(\K)$ be points of order $2$ of the form $T_1=(1:0:0)$, $T_2=(0:1:0)$ and $T_3=(a:b:0)=T_1+T_2$ on a Huff's curve $E_{Hu}$ given by the equation (\ref{HuProj}). For a finite point $P=(X:Y:Z)=(x,y)\ne(0,0)$ we consider the following translations:
\begin{equation}
\begin{array}{rl}
\tau_{T_1}\ :\ P+T_1 &= \left(\frac{1}{x},-y\right),\\
\tau_{T_2}\ :\ P+T_2 &= \left(-x,\frac{1}{y}\right).
\end{array}
\end{equation}
\begin{rem}\label{remHu8}
Now we intend to check if the above translations are projectively linear. As in Remark \ref{remHu4}, we consider the Huff's curve equation $E_{Hu}$  in the space $(\mathbb{P}^1)^2$ given by the equation (\ref{HuPP}).
In  the space $(\mathbb{P}^1)^2$ for point $P=((X:Z_1):(Y:Z_2))$ the translations $\tau_1$ and $\tau_2$ have the following form
\begin{equation}
  \begin{array}{rl}
  \tau_{T_1}\ :\ P+T_1 &= ((Z_1:X),(-Y:Z_2)),\\
  \tau_{T_2}\ :\ P+T_2 &= ((-X:Z_1),(Z_2:Y)).
  \end{array}
  \end{equation}
By embedding the above formulas into a projective space via Segre embedding $\rho  :
 \mathbb{P}^1\times \mathbb{P}^1\to \mathbb{P}^3$ given by (\ref{Segre}), we  get
 \begin{equation}
  \begin{array}{rl}
    P &= (XY:XZ_2:YZ_1:Z_1Z_2)=(U_1:U_2:U_3:U_4),\\
  \tau_{T_1}\ :\ P+T_1 &= (-YZ_1:Z_1Z_2:-XY:XZ_2)=(-U_3:U_4:-U_1:U_2),\\
  \tau_{T_2}\ :\ P+T_2 &= (-XZ_2:-XY:Z_1Z_2:YZ_1)=(-U_2:-U_1:U_4:U_3).
  \end{array}
  \end{equation}
    Additionally, we have
\begin{equation}
 -P = (XY:-XZ_2:-YZ_1:Z_1Z_2)=(U_1:-U_2:-U_3:U_4).
\end{equation}
  In consequence, we see that the translations $\tau_1$, $\tau_2$ and the involution $[-1]P$ are projectively linear in $\mathbb{P}^3$.
\end{rem}

We will give an example of a compression function of degree $8$ for a Huff's curve $E_{Hu}: ax(y^2-1)=by(x^2-1)$.
A compression function of degree $8$ may be given by $f_8(x,y)=xy+\frac{1}{xy}-\frac{x}{y}-\frac{y}{x}$. It is easy to show, that $f_8$ has the degree of $8$. At first, let us note, that
\begin{itemize}
\item involution $[-1]P$ is projectively linear in $\mathbb{P}^3$ (see Remark \ref{remHu8}),
\item points $T_1=(1:0:0)$ and $T_2=(0:1:0)$ of order $2$ naturally belong to $E_{Hu}(\K)$,
\item translations $\tau_{T_1} :E_E \to E_E: \tau_{T_1}(P)=P+T_1$ and $\tau_{T_2} :E_E \to E_E: \tau_{T_2}(P)=P+T_2$ are also projectively linear (see Remark \ref{remHu8}).
\end{itemize}
Let us note that $f_{8}(P)=f_8(Q)$, iff $Q=\pm P + [l](1:0:0)+[k](0:1:0)$, for $l,k=\ov{0,1}$ and $P=(x,y)$ is in set\\$S=\{(x,y), (-x,-y), (\frac{1}{x},-y), (-\frac{1}{x},y), (-x_,\frac{1}{y}), (x,-\frac{1}{y}), (\frac{1}{x},\frac{1}{y}), (-\frac{1}{x},-\frac{1}{y})\}$.

\medskip
Let us assume that $xy=t$ and $t\ne 0$. From the Huff's curve equation (\ref{Hu}), one may derive:
\begin{equation}
 x^2=\frac{xy(axy+b)}{bxy+a}=\frac{t(at+b)}{bt+a},\quad y^2=\frac{xy(bxy+a)}{axy+b}=\frac{t(bt+a)}{at+b}.
\end{equation}
Then, one may write
\begin{equation}
 f_8(x,y)=xy+\frac{1}{xy}-\frac{x}{y}-\frac{y}{x}=t+\frac{1}{t}-\frac{at+b}{bt+a}-\frac{bt+a}{at+b}.
\end{equation}

Let us denote $f_8(x,y)=r$. By simple calculations, we get
\begin{equation}
\label{Hdeg8}
h(t) =  a b t^{4} -a b r t^{3} - (a^2r+b^2r+2ab)t^2 - a b r t  + a b=0.
\end{equation}
Polynomial $h(t)$ has at most $4$ distinct roots. Let $t_1$ be a root of $h(t)$. If we substitute $y=\frac{t_1}{x}$ in the Huff's curve equation (\ref{Hu}), we get the quadratic equation (\ref{Huxy}). In consequence, a degree of the compression function $f_8$ is 8 at most.
\begin{theorem}
A differential addition formula for $f_{8}(P+Q)f_{8}(P-Q)$ on a Huff's curve and a compression function $f_8(x,y)=xy+\frac{1}{xy}-\frac{x}{y}-\frac{y}{x}$ is given by:

\small
\begin{equation}
\label{daddhu8}
f_{8}(P+Q)f_{8}(P-Q)=\frac{(r_Pr_Q-16)^2}{(r_P-r_Q)^2}.
\end{equation}
\normalsize
Similarly, doubling is given by
\small
\begin{equation}
\label{doubhu8}
f_{8}([2]P)=\frac{(r_P^2-16)^2}{4r_P(r_P^2+4\frac{a^2+b^2}{ab}r_P+16)}.
\end{equation}
\end{theorem}
\normalsize

\begin{expl}
Formula (\ref{daddhu8}) may be obtained using the algorithm from Appendix \ref{appDiffMain}, with modifications from Appendix \ref{appDiff8HU}. Correspondingly, formula (\ref{doubhu8}) may be obtained using the algorithm from Appendix \ref{appDoubMain}, with modifications from Appendix \ref{appDoub8HU}.
\end{expl}

\begin{rem}
Let us note that formulas (\ref{daddhu8}) and (\ref{doubhu8}) that we obtained for compression function $f_{8}$ are as efficient as formulas for the Montgomery curve.
\end{rem}

\subsubsection{Compression function of degree 16}

One may check, in a similar manner as in the preceding sections, that a compression function of degree 16 is given by
\small
\begin{equation}
\label{fcdeg16}
f_{16}(x,y)=xy+\frac{1}{xy}-\frac{y}{x}-\frac{x}{y}+\frac{y+1}{1-y} \cdot \frac{x+1}{1-x}+\frac{y+1}{y-1} \cdot \frac{1-x}{1+x}+\frac{y-1}{1+y} \cdot \frac{x-1}{x+1}+\frac{1-y}{1+y} \cdot \frac{x+1}{x-1}.
\end{equation}
\normalsize
This compression function may be obtained using natural symmetries on Huff's curves and action on three $2$-torsion points and points of order of 4, given by $(\pm 1: \pm 1: 1)$.

\medskip
Let us note that $f_{16}(P)=f_{16}(Q)$, iff $Q=\pm P + [l](1:0:0)+[m](1:1:1)$, for $l=\ov{0,1}, m=\ov{0,3}$ and $P=(x,y)$ is in set
\small
\begin{equation*}
\begin{array}{lr}
S=\Big\{
\left(x, y \right),
\left(-x, -y\right),
\left(\frac{1}{x}, -y\right),
\left(-\frac{1}{x}, y\right),
\left(x, -\frac{1}{y}\right),{}
\left(-x, \frac{1}{y}\right),
\left(\frac{1}{x}, \frac{1}{y}\right),
\left(-\frac{1}{x}, -\frac{1}{y}\right),\\
\left(\frac{y+1}{1-y}, \frac{x+1}{1-x}\right),
\left(\frac{y+1}{y-1}, \frac{1-x}{1+x}\right),
\left(\frac{y-1}{1+y}, \frac{x-1}{x+1}\right),
\left(\frac{1-y}{1+y}, \frac{x+1}{x-1}\right),
\left(-\frac{y+1}{1-y}, -\frac{x+1}{1-x}\right),
\left(-\frac{y+1}{y-1}, -\frac{1-x}{1+x}\right),\\
\left(-\frac{y-1}{1+y}, -\frac{x-1}{x+1}\right),
\left(-\frac{1-y}{1+y}, -\frac{x+1}{x-1}\right)
\Big\}.
\end{array}
\end{equation*}
\normalsize

\begin{theorem}

The differential addition formula for $f_{16}(P+Q)f_{16}(P-Q)$ on a Huff's curve and a compression function $f_{16}(x,y)$ given by the formulae \eqref{fcdeg16}, where $f_{16}(P)=r_P$ and $f_{16}(Q)=r_Q$, is given by:

\small
\begin{equation}
\label{daddhu16}
\begin{array}{rl}
f_{16}(P+Q) f_{16}(P-Q)=\frac{(r_Pr_Q+64)^2+1024\frac{a^2+b^2}{ab}(r_P+r_Q+\frac{a^2+b^2}{ab})}{(r_P-r_Q)^2}.
\end{array}
\end{equation}
\normalsize
Similarly, doubling is given by
\small
\begin{equation}
f([2]P)=\frac{L}{M},
\end{equation}
where
\begin{equation}
\label{doubhu16}
\begin{array}{lr}
L=\frac{1}{4} x^4 + 32 x^2 + \frac{512 a^2 + 512 b^2}{a b} x + \frac{1024 a^4 + 3072 a^2*b^2 + 1024 b^4}{a^2 b^2},\\
M=x^3 + \frac{4 a^2 + 4 b^2}{a b} x^2 - 64 x + \frac{-256 a^2 - 256 b^2}{ab}.
\end{array}
\end{equation}
\end{theorem}

\begin{expl}
Formula (\ref{daddhu16}) may be obtained using the algorithm from Appendix \ref{appDiffMain}, with modifications from Appendix \ref{appDiff16HU}. Correspondingly, formula (\ref{doubhu16}) may be obtained using the algorithm from Appendix \ref{appDoubMain}, with modifications from Appendix \ref{appDoub16HU}.
\end{expl}

\section{Formulas as fast as Montgomery's}
A short analysis of the cost of applying a compression functions $f_2, f_6$ and $f_{18}$ on a generalized Hessian curve shows that the applications of these functions are not as efficient as Montgomery-like formulas for Montgomery, Huff's and Edwards curves. Now, we will present the following theorem.

\begin{thm}
\label{themIsoMon}
Let $E$ be a model of an elliptic curve, for which isomorphism $\phi$ from $E$ to the Montgomery curve $E_M:B\ov{y}^2=\ov{x}^3+A\ov{x}^2+\ov{x}$ is given by a function $\phi(x,y)=\left( W_x(x,y), W_y(x,y) \right)$, where $W_x(x,y), W_y(x,y)$ are rational functions. Let us $f_2(\ov{P})=\ov{x}$ be compression function of degree $2$ on the Montgomery curve, where $\ov{P}=(\ov{x}, \ov{y}) \in E_M(\K)$. Then $g_2(x,y)=f_2\left( W_x(x,y) \right)$ is compression function of degree $2$. Let $A(f_2(\ov{P}),f_2(\ov{Q}),f_2(\ov{P}-\ov{Q}))$ be differential addition and $D(f_2(\ov{P}))$ be the doubling formulas on the Montgomery curve. In such a case, on an elliptic curve $E$ we may define differential addition $A(g_2(P),g_2(Q),g_2(P-Q))$ and doubling $D(g_2(P))$ formulas of the same efficiency as Montgomery's, up to constants which depends on the coefficients of $E$.
\end{thm}

\begin{proof}
Let us $\phi$ be an isomorphism from a curve $E$ to the Montgomery curve $E_M:B\ov{y}^2=\ov{x}^3+A\ov{x}^2+\ov{x}$, given by $\phi(P)=\left( W_x(P), W_y(P) \right)$, where $W_x(P)$ and $W_y(P)$ are rational functions. Then, for $P \in E(\K)$ holds that $g_2(P)=f_2\left( W_x(P), W_y(P) \right)=W_x(P)$ is indeed a compression function of degree $2$ on a curve $E$. Let us note, that $f_2(\ov{P})=\ov{x}=\ov{r}_{P}$ gives the same value $\ov{r}_{P}$ only for two points $\ov{P}, -\ov{P} \in E_M(\K)$. Because $\phi$ is an isomorphism from $E$ to $E_M$, it means that $f_2\left( W_x(P), W_y(P) \right)=W_x(P)=r_P$ also gives the same value $r_P$ for only two points $P, -P \in E(\K)$. It follows that $g_2(P)=f_2\left( W_x(P), W_y(P) \right)=W_x(P)$ is a compression function of degree $2$ on a curve $E$.

Now, let us denote $\ov{r}_{\ov{P}}=f_2(\ov{P}), \ov{r}_{\ov{Q}}=f_2(\ov{Q}), \ov{r}_{\ov{P}+\ov{Q}}=f_2(\ov{P}+\ov{Q}), \ov{r}_{\ov{P}-\ov{Q}}=f_2(\ov{P}-\ov{Q}), \ov{r}_{[2]\ov{P}}=f_2([2]\ov{P}) $. Let us note that hold $r_P=g_2(P), r_Q=g_2(Q), r_{P+Q}=g_2(P+Q), r_{P-Q}=g_2(P-Q), r_{[2]P}=g_2([2]P)$. If $A(f_2(\ov{P}),f_2(\ov{Q}),f_2(\ov{P}-\ov{Q}))$ is a differential addition and $D(f_2(\ov{P}))$ is a doubling formula on the Montgomery curve, then $A( f_2\left( W_x(P), W_y(P) \right), f_2\left( W_x(Q), W_y(Q) \right),\linebreak f_2\left( W_x(P-Q), W_y(P-Q) \right) )= A\left( g_2(P), g_2(Q), g_2(P-Q) \right)$ is a differential addition formula on a curve $E$. Correspondingly, if $D(f_2(\ov{P})$ is a doubling formula on the Montgomery curve, then $D(f_2(\ov{P})= D(f_2\left( W_x(P), W_y(P) \right))=D(g_2(P))$ is a doubling formula on a curve $E$. Because for every $\ov{P} \in E_M(\K)$ holds that $r_{\ov{P}}=f_2(\ov{P})=W_x(P)=g_2(P)=r_P$, it follows that $A(r_P, r_Q, r_{P-Q})$ and $D(r_P)$ on $E$ are of the same efficiency as $A(\ov{r}_{\ov{P}},\ov{r}_{\ov{Q}},\ov{r}_{\ov{P}-\ov{Q}})$ and $D(\ov{r}_{\ov{P}})$ on the Montgomery curve, up to constants which depends on the coefficients of a curve $E$.
\end{proof}

Let us note that an isomorphism $\phi$ for which conditions presented in Theorem \ref{themIsoMon} hold, may be defined, for example, from a twisted Edwards curve to the Montgomery \cite{Ber08}, from the Huff's curves to the Montgomery \cite{Far10} and also other models of elliptic curves. Using these isomorphisms, one may obtain for these compression functions of degree $2$ Montgomery-like formulas of the same efficiency. However, compression functions of degree $2$ obtained by Theorem \ref{themIsoMon} may be sometimes complicated, as same as constants appearing in differential addition and doubling formulas and therefore may be not optimal for all applications.

\begin{rem}
Let us note that for Hessian, twisted Hessian and generalized Hessian curve models there do not exist natural isomorphisms from these curves to the Montgomery curve. We therefore state the conjecture that for Hessian, twisted Hessian and generalized Hessian curves arithmetics using compression of the same or similar efficiency as for Montgomery do not exist.
\end{rem}

\section{Conclusion}

This paper presents new compression functions of degree $> 2$ on Edwards, Huff's and the Hessian family of elliptic curves. As it was presented in section \ref{sec2}, compression functions of high degree may be obtained using natural symmetries on elliptic curves obtained by the action on the $n$-torsion
 point~$T$.

Additionally, it is worth noteing that models of elliptic curves for which a birationally equivalent Montgomery curve exists, have some compression functions of degree $2$ for which differential addition and doubling is of the same efficiency as the Montgomery curves. Unfortunately, it seems that compression functions of the same efficiency as the Montgomery curve do not exist for models of elliptic curves with a natural point of order $3$. Such representatives of elliptic curves are Hessian, twisted Hessian, and generalized Hessian curves. This is because, for these models, natural isomorphisms do not exist from these curves to the Montgomery curve.

 \bigskip

\begin{subappendices}
\renewcommand{\thesection}{\Alph{section}}

\section{Computer program for generation differential addition formula}
\subsection{Main program}
\label{appDiffMain}
\begin{verbatim}
Q:=Rationals();
Z:=Integers();
rQ<a,b>:=FunctionField(Q,2);
/* maximal degree d of nominator and denominator in rational function
   for f(P+Q)*f(P-Q) or for f(P+Q)+f(P-Q) */
d:=4;
n:=(d+1)*(d+2);     /* Number of unknown parameters   */
R:=PolynomialRing(rQ,n);
F:= FieldOfFractions(R);
pF2<x,y>:=PolynomialRing(F,2);
pF4<x1,y1,x2,y2>:=PolynomialRing(F,4);
rF4:=FieldOfFractions(pF4);


//Beginning of exchangeable parameters
/*definition of an elliptic curve E*/
E:=x^3 + y^3 +a-b*x*y;
/*definition of compression function f*/
f:=x*y;
/* addition of points (x1,y1)+(x2,y2), depends on the curve equation */
x3:=(y1^2*x2-y2^2*x1)/(x2*y2-x1*y1);
y3:=(x1^2*y2-x2^2*y1)/(x2*y2-x1*y1);
/* subtraction of points (x1,y1)-(x2,y2), depends on the curve equation */
x4:=Evaluate(x3, [x1,y1,y2,x2]);
y4:=Evaluate(y3, [x1,y1,y2,x2]);
//End of exchangeable parameters

I:=ideal<pF4|[ Evaluate(E,[x1,y1]), Evaluate(E, [x2,y2]) ] >;
f1:=Evaluate(f,[x1,y1]);  f2:=Evaluate(f,[x2,y2]);
f3:=Evaluate(f,[x3,y3]);  f4:=Evaluate(f,[x4,y4]);

/* In here we search for rational function f(P+Q)*f(P-Q).
   If one intends to search for rational function f(P+Q)+f(P-Q),
                                                      then H:=f3+f4; */
H:=f3*f4;						
G:=[pF2!0,pF2!0];
k:=0;
for u:=1 to 2 do
  for  j:=1 to d+1 do
    for i:=1 to j do  k:=k+1;
      G[u]:=G[u]+ R.k*x^(i-1)*y^(j-i);
end for; end for; end for;
Nor:=NormalForm(Numerator( H - Evaluate(G[1]/G[2],[f1,f2])), I);
cf:=Coefficients(Nor);
sd:=[];
for i in cf do sd:= sd cat [Denominator(i)]; end for;
ld:=Lcm(sd);
cf0:=[];
/* multiplication by common denominator Coefficients(Nor) */
for  i:=1 to #cf do  cf0:=cf0 cat [R!(ld*cf[i])]; end for;
Proj:=ProjectiveSpace(R);
Sch:=Scheme(Proj,cf0);
dim:=Dimension(Sch);
if dim eq 0 then Rp:=RationalPoints(Sch);
  for i in Rp do sq:=Eltseq(i); end for;
  G:=[pF2!0,pF2!0];
  k:=0;
  for u:=1 to 2 do
    for  j:=1 to d+1 do
      for i:=1 to j do  k:=k+1;
        G[u]:=G[u]+ sq[k]*x^(i-1)*y^(j-i);
  end for; end for; end for;
  [G[1], G[2]];
end if;
\end{verbatim}

\subsection{Modifications for compression function $f_{18}(x,y)=xy$ on generalized Hessian curve}
\label{appDiff2MultGH}
\begin{verbatim}
//Beginning of exchangeable parameters
/*definition of an elliptic curve E*/
E:=x^3 + y^3 +1-b*x*y;
/*definition of compression function f*/
f:=x+y;
/* addition of points (x1,y1)+(x2,y2), depends on the curve equation */
x3:=(y1^2*x2-y2^2*x1)/(x2*y2-x1*y1);
y3:=(x1^2*y2-x2^2*y1)/(x2*y2-x1*y1);
x4:=Evaluate(x3, [x1,y1,y2,x2]);
y4:=Evaluate(y3, [x1,y1,y2,x2]);
//End of exchangeable parameters
\end{verbatim}

\subsection{Modifications for compression function $f_{18}(x,y)=xy$ on generalized Hessian curve}
\label{appDiff2AddGH}
\begin{verbatim}
//Beginning of exchangeable parameters
/*definition of an elliptic curve E*/
E:=x^3 + y^3 +1-b*x*y;
/*definition of compression function f*/
f:=x+y;
/* addition of points (x1,y1)+(x2,y2), depends on the curve equation */
x3:=(y1^2*x2-y2^2*x1)/(x2*y2-x1*y1);
y3:=(x1^2*y2-x2^2*y1)/(x2*y2-x1*y1);
x4:=Evaluate(x3, [x1,y1,y2,x2]);
y4:=Evaluate(y3, [x1,y1,y2,x2]);
//End of exchangeable parameters
\end{verbatim}
Additionally, $H=f3+f4$.

\subsection{Modifications for compression function $f_{18}(x,y)=xy$ on generalized Hessian curve}
\label{appDiff18GH}
\begin{verbatim}
//Beginning of exchangeable parameters
/*definition of an elliptic curve E*/
E:=x^3 + y^3 +1-b*x*y;
/*definition of compression function f*/
f:=(x^3*y^3+x^3+y^3)/(x^2*y^2);
/* addition of points (x1,y1)+(x2,y2), depends on the curve equation */
x3:=(y1^2*x2-y2^2*x1)/(x2*y2-x1*y1);
y3:=(x1^2*y2-x2^2*y1)/(x2*y2-x1*y1);
x4:=Evaluate(x3, [x1,y1,y2,x2]);
y4:=Evaluate(y3, [x1,y1,y2,x2]);
//End of exchangeable parameters
\end{verbatim}

\subsection{Modifications for compression function $f_2(x,y)=y$ on Edwards curve}
\label{appDiff2ED}
\begin{verbatim}
//Beginning of exchangeable parameters
/*definition of an elliptic curve E*/
a:=1;
E:=a*x^2 + y^2 -1 - b*x^2*y^2;
/*definition of compression function f*/
f:=y;
/* addition of points (x1,y1)+(x2,y2), depends on the curve equation */
x3:=(x1*y2+y1*x2)/(1+b*x1*x2*y1*y2);
y3:=(y1*y2-a*x1*x2)/(1-b*x1*x2*y1*y2);
/* subtraction of points (x1,y1)-(x2,y2), depends on the curve equation */
x4:=Evaluate(x3, [x1,y1,-x2,y2]);
y4:=Evaluate(y3, [x1,y1,-x2,y2]);
//End of exchangeable parameters
\end{verbatim}

\subsection{Modifications for compression function $f_4(x,y)=y^2$ on Edwards curve}
\label{appDiff4ED}
\begin{verbatim}
//Beginning of exchangeable parameters
/*definition of an elliptic curve E*/
a:=1;
E:=a*x^2 + y^2 -1 - b*x^2*y^2;
/*definition of compression function f*/
f:=y^2;
/* addition of points (x1,y1)+(x2,y2), depends on the curve equation */
x3:=(x1*y2+y1*x2)/(1+b*x1*x2*y1*y2);
y3:=(y1*y2-a*x1*x2)/(1-b*x1*x2*y1*y2);
/* subtraction of points (x1,y1)-(x2,y2), depends on the curve equation */
x4:=Evaluate(x3, [x1,y1,-x2,y2]);
y4:=Evaluate(y3, [x1,y1,-x2,y2]);
//End of exchangeable parameters
\end{verbatim}

\subsection{Modifications for compression function $f_8(x,y)=x^2 y^2$ on Edwards curve}
\label{appDiff8ED}
\begin{verbatim}
//Beginning of exchangeable parameters
/*definition of an elliptic curve E*/
a:=1;
E:=a*x^2 + y^2 -1 - b*x^2*y^2;
/*definition of compression function f*/
f:=x^2*y^2;
/* addition of points (x1,y1)+(x2,y2), depends on the curve equation */
x3:=(x1*y2+y1*x2)/(1+b*x1*x2*y1*y2);
y3:=(y1*y2-a*x1*x2)/(1-b*x1*x2*y1*y2);
/* subtraction of points (x1,y1)-(x2,y2), depends on the curve equation */
x4:=Evaluate(x3, [x1,y1,-x2,y2]);
y4:=Evaluate(y3, [x1,y1,-x2,y2]);
//End of exchangeable parameters
\end{verbatim}

\subsection{Modifications for compression function $f_2(x,y)=xy$ on Huff's curve}
\label{appDiff2HU}
\begin{verbatim}
//Beginning of exchangeable parameters
/*definition of an elliptic curve E*/
E:=a*x*(y^2-1)-b*y*(x^2-1);
/*definition of compression function f*/
f:=x*y;
/* addition of points (x1,y1)+(x2,y2), depends on the curve equation */
x3:=(x1+x2)*(y1*y2+1)/((x1*x2+1)*(1-y1*y2));
y3:=(y1+y2)*(x1*x2+1)/((1-x1*x2)*(y1*y2+1));
/* subtraction of points (x1,y1)-(x2,y2), depends on the curve equation */
x4:=Evaluate(x3, [x1,y1,-x2,-y2]);
y4:=Evaluate(y3, [x1,y1,-x2,-y2]);
//End of exchangeable parameters
\end{verbatim}

\subsection{Modifications for compression function $f_4(x,y)=xy+\frac{1}{xy}$ on Huff's curve}
\label{appDiff4HU}
\begin{verbatim}
//Beginning of exchangeable parameters
/*definition of an elliptic curve E*/
E:=a*x*(y^2-1)-b*y*(x^2-1);
/*definition of compression function f*/
f:=x*y+1/(x*y);
/* addition of points (x1,y1)+(x2,y2), depends on the curve equation */
x3:=(x1+x2)*(y1*y2+1)/((x1*x2+1)*(1-y1*y2));
y3:=(y1+y2)*(x1*x2+1)/((1-x1*x2)*(y1*y2+1));
/* subtraction of points (x1,y1)-(x2,y2), depends on the curve equation */
x4:=Evaluate(x3, [x1,y1,-x2,-y2]);
y4:=Evaluate(y3, [x1,y1,-x2,-y2]);
//End of exchangeable parameters
\end{verbatim}

\subsection{Modifications for compression function $f_8(x,y)=xy+\frac{1}{xy}-\frac{x}{y}-\frac{y}{x}$ on Huff's curve}
\label{appDiff8HU}
\begin{verbatim}
//Beginning of exchangeable parameters
/*definition of an elliptic curve E*/
E:=a*x*(y^2-1)-b*y*(x^2-1);
/*definition of compression function f*/
f:=x*y+1/(x*y)-x/y-y/x;
/* addition of points (x1,y1)+(x2,y2), depends on the curve equation */
x3:=(x1-x2)*(y1+y2)/((y1-y2)*(1-x1*x2));
y3:=(y1-y2)*(x1+x2)/((x1-x2)*(1-y1*y2));
/* subtraction of points (x1,y1)-(x2,y2), depends on the curve equation */
x4:=Evaluate(x3, [x1,y1,-x2,-y2]);
y4:=Evaluate(y3, [x1,y1,-x2,-y2]);
//End of exchangeable parameters
\end{verbatim}

\subsection{Modifications for compression function $f_{16}(x,y)=$ given by equation (\ref{fcdeg16}) on Huff's curve}
\label{appDiff16HU}
\begin{verbatim}
//Beginning of exchangeable parameters
/*definition of an elliptic curve E*/
E:=a*x*(y^2-1)-b*y*(x^2-1);
/*definition of compression function f*/
f:=x*y+1/(x*y)-y/x-x/y+(y+1)/(1-y)*(x+1)/(1-x)+(y+1)/(y-1)*(1-x)/(1+x)+
(y-1)/(1+y)*(x-1)/(x+1)+(1-y)/(1+y)*(x+1)/(x-1);
/* addition of points (x1,y1)+(x2,y2), depends on the curve equation */
x3:=(x1-x2)*(y1+y2)/((y1-y2)*(1-x1*x2));
y3:=(y1-y2)*(x1+x2)/((x1-x2)*(1-y1*y2));
/* subtraction of points (x1,y1)-(x2,y2), depends on the curve equation */
x4:=Evaluate(x3, [x1,y1,-x2,-y2]);
y4:=Evaluate(y3, [x1,y1,-x2,-y2]);
//End of exchangeable parameters
\end{verbatim}

\smallskip
\section{Doubling}
\subsection{Main program}
\label{appDoubMain}
\begin{verbatim}
Q:= Rationals();
Z:=Integers();
rQ<a,b>:=FunctionField(Q,2);
for d  in [1..10] do
n:=2*d+2;
R:=PolynomialRing(rQ,n);
F:= FieldOfFractions(R);
pF2<x,y>:=PolynomialRing(F,2);
rF2:=FieldOfFractions(pF2);
pF4<x1,y1,x2,y2>:=PolynomialRing(F,4);
rF4:=FieldOfFractions(pF4);


//Beginning of exchangeable parameters
/*definition of an elliptic curve E*/
E:=x^3 + y^3 +a-b*x*y;
/*definition of compression function f*/
f:=x*y;
/* addition of points (x1,y1)+(x2,y2), depends on the curve equation */
x3 := (a*y1-x2*y2*x1^2)/(x1*x2^2-y2*y1^2);
y3 := (x1*y1*y2^2 - a*x2)/(x1*x2^2-y2*y1^2);
//End of exchangeable parameters


I:=ideal<pF2|E>;
D1:=Evaluate(x3,[x,y,x,y]);   /*doubling */
D2:=Evaluate(y3,[x,y,x,y]);
H1:=D1; H2:=D2;
F1:=0; F2:=0;
for  j:=1 to n do
  if j le  d+1 then
    F1:=F1+R.j*x^(j-1);
  else
    F2:=F2+R.j*x^(j-d-2);
  end if;
end for;


Nor:=NormalForm(Numerator(Evaluate(f,[H1,H2]) - Evaluate(F1/F2,[f,1])), I);
cf:=Coefficients(Nor);    sd:=[];
for i in cf do sd:= sd cat [Denominator(i)]; end for;
ld:=Lcm(sd);
cf0:=[];
for  i:=1 to #cf do  cf0:=cf0 cat [R!(ld*cf[i])];
end for;
Proj:=ProjectiveSpace(R);
Sch:=Scheme(Proj,cf0);
dim:=Dimension(Sch);
if dim eq 0 then Rp:=RationalPoints(Sch);
for i in Rp do sq:=Eltseq(i); end for;
F1:=0; F2:=0;
for  j:=1 to n do
  if j le  d+1 then
    F1:=F1+sq[j]*x^(j-1);
  else
    F2:=F2+sq[j]*x^(j-d-2);
  end if;
end for;
[F1, F2];
break d;
end if;
end for;
\end{verbatim}

\subsection{Modifications for compression function $f_{2}(x,y)=xy$ on generalized Hessian curve}
\label{appDoub2GH}
\begin{verbatim}
//Beginning of exchangeable parameters
/*definition of an elliptic curve E*/
E:=x^3 + y^3 +1-b*x*y;
/*definition of compression function f*/
f:=x+y;
/* addition of points (x1,y1)+(x2,y2), depends on the curve equation */
x3 := (y1-x2*y2*x1^2)/(x1*x2^2-y2*y1^2);
y3 := (x1*y1*y2^2 - x2)/(x1*x2^2-y2*y1^2)
//End of exchangeable parameters
\end{verbatim}

\subsection{Modifications for compression function $f_{18}(x,y)=xy$ on generalized Hessian curve}
\label{appDoub18GH}
\begin{verbatim}
//Beginning of exchangeable parameters
/*definition of an elliptic curve E*/
E:=x^3 + y^3 +1-b*x*y;
/*definition of compression function f*/
f:=(x^3*y^3+x^3+y^3)/(x^2*y^2);
/* addition of points (x1,y1)+(x2,y2), depends on the curve equation */
x3 := (y1-x2*y2*x1^2)/(x1*x2^2-y2*y1^2);
y3 := (x1*y1*y2^2 - x2)/(x1*x2^2-y2*y1^2)
//End of exchangeable parameters
\end{verbatim}

\subsection{Modifications for compression function $f_2(x,y)=y$ on Edwards curve}
\label{appDoub2ED}
\begin{verbatim}
//Beginning of exchangeable parameters
/*definition of an elliptic curve E*/
a:=1;
E:=a*x^2 + y^2 -1 - b*x^2*y^2;
/*definition of compression function f*/
f:=y;
/* addition of points (x1,y1)+(x2,y2), depends on the curve equation */
x3:=(x1*y2+y1*x2)/(1+b*x1*x2*y1*y2);
y3:=(y1*y2-a*x1*x2)/(1-b*x1*x2*y1*y2);
//End of exchangeable parameters
\end{verbatim}

\subsection{Modifications for compression function $f_4(x,y)=y^2$ on Edwards curve}
\label{appDoub4ED}
\begin{verbatim}
//Beginning of exchangeable parameters
/*definition of an elliptic curve E*/
a:=1;
E:=a*x^2 + y^2 -1 - b*x^2*y^2;
/*definition of compression function f*/
f:=y^2;
/* addition of points (x1,y1)+(x2,y2), depends on the curve equation */
x3:=(x1*y2+y1*x2)/(1+b*x1*x2*y1*y2);
y3:=(y1*y2-a*x1*x2)/(1-b*x1*x2*y1*y2);
//End of exchangeable parameters
\end{verbatim}

\subsection{Modifications for compression function $f_8(x,y)=x^2 y^2$ on Edwards curve}
\label{appDoub8ED}

\begin{verbatim}
//Beginning of exchangeable parameters
/*definition of an elliptic curve E*/
a:=1;
E:=a*x^2 + y^2 -1 - b*x^2*y^2;
/*definition of compression function f*/
f:=x^2*y^2;
/* addition of points (x1,y1)+(x2,y2), depends on the curve equation */
x3:=(x1*y2+y1*x2)/(1+b*x1*x2*y1*y2);
y3:=(y1*y2-a*x1*x2)/(1-b*x1*x2*y1*y2);
//End of exchangeable parameters
\end{verbatim}

\eject

\subsection{Modifications for compression function $f_2(x,y)=xy$ on Huff's curve}
\label{appDoub2HU}
\begin{verbatim}
//Beginning of exchangeable parameters
/*definition of an elliptic curve E*/
E:=a*x*(y^2-1)-b*y*(x^2-1);
/*definition of compression function f*/
f:=x*y;
/* addition of points (x1,y1)+(x2,y2), depends on the curve equation */
x3:=(x1+x2)*(y1*y2+1)/((x1*x2+1)*(1-y1*y2));
y3:=(y1+y2)*(x1*x2+1)/((1-x1*x2)*(y1*y2+1));
//End of exchangeable parameters
\end{verbatim}

\subsection{Modifications for compression function $f_4(x,y)=xy+\frac{1}{xy}$ on Huff's curve}
\label{appDoub4HU}
\begin{verbatim}
//Beginning of exchangeable parameters
/*definition of an elliptic curve E*/
E:=a*x*(y^2-1)-b*y*(x^2-1);
/*definition of compression function f*/
f:=x*y+1/(x*y);
/* addition of points (x1,y1)+(x2,y2), depends on the curve equation */
x3:=(x1+x2)*(y1*y2+1)/((x1*x2+1)*(1-y1*y2));
y3:=(y1+y2)*(x1*x2+1)/((1-x1*x2)*(y1*y2+1));
//End of exchangeable parameters
\end{verbatim}

\subsection{Modifications for compression function $f_8(x,y)=xy+\frac{1}{xy}-\frac{x}{y}-\frac{y}{x}$ on Huff's curve}
\label{appDoub8HU}
\begin{verbatim}
//Beginning of exchangeable parameters
/*definition of an elliptic curve E*/
E:=a*x*(y^2-1)-b*y*(x^2-1);
/*definition of compression function f*/
f:=x*y+1/(x*y)-x/y-y/x;
/* doubling of point (x1,y1), depends on the curve equation */
x3:=(2*y1^2+2)*x1/((x1^2+1)*y1^2-x1^2-1);
y3:=(2*x1^2+2)*y1/((x1^2-1)*y1^2+x1^2-1);
//End of exchangeable parameters
\end{verbatim}

\subsection{Modifications for compression function $f_{16}(x,y)$ given by equation (\ref{fcdeg16}) on Huff's curve}
\label{appDoub16HU}
\begin{verbatim}
//Beginning of exchangeable parameters
/*definition of an elliptic curve E*/
E:=a*x*(y^2-1)-b*y*(x^2-1);
/*definition of compression function f*/
f:=x*y+1/(x*y)-y/x-x/y+(y+1)/(1-y)*(x+1)/(1-x)+(y+1)/(y-1)*(1-x)/(1+x)+
(y-1)/(1+y)*(x-1)/(x+1)+(1-y)/(1+y)*(x+1)/(x-1);
/* doubling of point (x1,y1), depends on the curve equation */
x3:=(2*y1^2+2)*x1/((x1^2+1)*y1^2-x1^2-1);
y3:=(2*x1^2+2)*y1/((x1^2-1)*y1^2+x1^2-1);
//End of exchangeable parameters
\end{verbatim}

\end{subappendices}

\end{document}